\documentclass[11pt,epsf]{article}

\long\def\comment#1{\ifdim\overfullrule>0pt{\sf[{#1}]}\fi}

\usepackage{comment,amsmath,amsthm,graphicx}
\usepackage{latexsym}
\usepackage{epsfig}
\usepackage{bbm}
\usepackage[usenames]{color}
\usepackage{colordvi}
\usepackage{pdfsync}
\usepackage{amsfonts}
\usepackage{color}
\usepackage{appendix}
\usepackage{enumerate}

\newtheorem{theorem}{Theorem}
\newtheorem{definition}{Definition}

\newtheorem{claim}{Claim} 
\newtheorem{lemma}{Lemma}
\newtheorem{fact}{Fact}
\newtheorem{corollary}{Corollary}

\ifx\pdftexversion\undefined
\usepackage[colorlinks,linkcolor=black,filecolor=black,citecolor=black,urlco
lor=black,pdfstartview=FitH]{hyperref}
\else
\usepackage[colorlinks,linkcolor=blue,filecolor=blue,citecolor=blue,urlcolor
=blue,pdfstartview=FitH]{hyperref}
\fi
\newcommand{\optLP}{{${OPT}_{LP}(G)$}}

\newcommand{\namedref}[2]{\hyperref[#2]{#1~\ref*{#2}}}

\newcommand{\alert}[1]{\textbf{\color{red}
[[[#1]]]}\marginpar{\textbf{\color{red}**}}\typeout{ALERT:
\the\inputlineno: #1}}

\newenvironment{cproof}
{\begin{proof}
 [Proof.]
 \vspace{-1.5\parsep}
}
{\renewcommand{\qed}{\hfill $\Diamond$} \end{proof}}

\begin{document}
\bibliographystyle{alpha}
\def\proofend{\hfill$\Box$\medskip}
\def\Proof{\noindent{\bf Proof:\ \ }}
\def\Sketch{\noindent{\bf Sketch:\ \ }}
\def\eps{\epsilon}
\def\xa{x_{max}}
\def\xi{x_{min}}
\def\xij{x_{min}(j)}
\def\xaj{x_{max}(j)}
\def\fij{f_{min}(j)}
\def\faj{f_{max}(j)}
\def\x{x}

\title{An improved analysis of the M\"omke-Svensson algorithm for
  graph-TSP on subquartic graphs\thanks{A preliminary version of these
results (with a worse approximation ratio) appeared in the
proceedings of the European Symposium on Algorithms 2014.}}

\author{Alantha Newman\thanks{CNRS-Universit\'e Grenoble Alpes and G-SCOP,
  F-38000 Grenoble, France.  
Supported in part
 by LabEx PERSYVAL-Lab (ANR--11-LABX-0025).  Email: {{\tt
      firstname.lastname@grenoble-inp.fr}.} 
}}

\maketitle

\begin{abstract}
M\"omke and Svensson presented a beautiful new approach for
the traveling salesman problem on a graph metric (graph-TSP), which
yields a $4/3$-approximation guarantee on subcubic graphs as
well as a substantial improvement over the $3/2$-approximation
guarantee of Christofides' algorithm on general graphs.  The crux of
their approach is to compute an upper bound on the minimum cost of a
circulation in a particular network, $C(G,T)$, where $G$ is the input
graph and $T$ is a carefully chosen spanning tree.  The cost of this
circulation is directly related to the number of edges in a tour
output by their algorithm.  Mucha subsequently improved the analysis
of the circulation cost, proving that M\"omke and Svensson's algorithm
for graph-TSP has an approximation ratio of at most $13/9$ on
general graphs.

This analysis of the circulation is local, and vertices with degree
four or five can contribute the most to its cost.  Thus,
hypothetically, there could exist a subquartic graph (a graph with
degree at most four at each vertex) for which Mucha's analysis of the
M\"omke-Svensson algorithm is tight.  We show that this
is not the case and that M\"omke and Svensson's algorithm for
graph-TSP has an approximation guarantee of at most $25/18$ on
subquartic graphs.  To prove this, we present different methods to
upper bound the minimum cost of a circulation on the network $C(G,T)$.
Our approximation guarantee holds for all graphs that have an
optimal solution for a standard linear programming relaxation of
graph-TSP with subquartic support.
\end{abstract}

\section{Introduction}

The metric traveling salesman problem (TSP) is one of the most
well-known problems in the field of combinatorial optimization and
approximation algorithms.  Given a complete graph, $G=(V,E)$, with
nonnegative edge weights that satisfy the triangle inequality, the
goal is to compute a minimum cost tour of $G$ that visits each vertex
exactly once.  Christofides' algorithm yields a tour with cost no more
than $3/2$ times that of an optimal tour~\cite{christofides1976worst}.
It remains a major open problem to improve upon this approximation
factor.

In the past few years, there have been many exciting developments
relating to {\em graph-TSP}.  In this setting, we are given an
unweighted graph $G=(V,E)$, and the goal is to find the shortest tour
that visits each vertex {\em at least} once.  This problem is
equivalent to the special case of metric TSP where the shortest path
distances in $G$ define the metric.  It is also equivalent to the
problem of finding a connected, spanning, Eulerian multigraph in $G$ with the
minimum number of edges.

A promising approach to improving upon the factor of $3/2$ for metric
TSP is to round a linear programming relaxation known as the Held-Karp
relaxation~\cite{held1970traveling}.  A lower bound of $4/3$ on its
integrality gap can be demonstrated using a family of graph-TSP
instances.  Since it is widely conjectured that its integrality gap is
also upper bounded by $4/3$, proving this for graph-TSP would be a
step towards a more comprehensive understanding of the relaxation and
would hopefully provide insights applicable to metric TSP.  However,
even in this special case of metric TSP, graph-TSP had also long
resisted significant progress before the recent spate of results.

\subsection{Recent progress on graph-TSP}\label{sec:prog}

In 2005, Gamarnik, Lewenstein and Sviridenko presented an algorithm
for graph-TSP on cubic, 3-edge-connected graphs with an approximation
factor of $3/2-5/389$~\cite{gamarnik2005improved}, thus proving that
Christofides' approximation guarantee of $3/2$ is not the best
possible factor for this class of graphs.  Their approach is based on
finding a cycle cover for which they can upper bound the number of
components.  This general approach was also taken by Boyd, Sitters,
van der Ster and Stougie who
combined it with polyhedral ideas to obtain approximation guarantees
of $4/3$ for cubic graphs and $7/5$ for subcubic graphs (i.e., graphs
with degree at most three at each vertex)~\cite{boyd2011tsp}.  Shortly
afterwards, Gharan, Saberi and Singh proved that a subtle modification
of Christofides' algorithm has an approximation guarantee of
$3/2-\eps_0$ for graph-TSP on general graphs, where $\eps_0$ is a
fixed constant with value approximately
$10^{-12}$~\cite{gharan2011randomized}.

M\"omke and Svensson then presented a beautiful new approach for
graph-TSP, which resulted in a substantial improvement over the
$3/2$-approximation guarantee of
Christofides~\cite{momke2016removing}.  Their approach
also leads to a simple $\frac{4}{3}$-approximation algorithm 
for subcubic graphs.  We will discuss their algorithm in more detail
in Section \ref{sec:MS}, since our paper is directly based on their
approach.  Ultimately, they were able to prove an approximation
guarantee of $1.461$ for graph-TSP.  Mucha subsequently gave an
improved analysis, thereby proving that M\"omke and Svensson's
algorithm for graph-TSP actually has an approximation ratio of at most
$13/9$~\cite{mucha2014frac}.  Seb\H{o} and Vygen introduced an
approach for graph-TSP based on ear decompositions and matroid
intersection, which incorporated the techniques of M\"omke and
Svensson, and improved the approximation ratio to $7/5$, where it
currently stands~\cite{sebHo2012shorter}.  For the special case of
$k$-regular graphs, Vishnoi gave an algorithm for graph-TSP with an
approximation guarantee that approaches 1 as $k$
increases~\cite{vishnoi2012permanent}.

\subsection{M\"omke-Svensson's approach to graph-TSP}\label{sec:MS}

Christofides' algorithm for graph-TSP finds a spanning tree of the
graph and adds to it a $J$-join, where $J$ is the set of vertices that
have odd degree in the spanning tree.\footnote{A $J$-join of a graph
  $G=(V,E)$ is a subgraph $F \subset E$ such that the degree of each
  vertex in $J \subseteq V$ is odd in $F$ and the degree of each
  vertex not in $J$ is even (and can be zero) in $F$.  An {\em
    odd-join} of $G$ is a $J$-join where $J = V$.  For example, a
  perfect matching in a cubic graph is an odd-join.}  Since the
spanning tree is connected, the resulting subgraph is clearly
connected, and since the $J$-join corrects the parity of the spanning
tree, the resulting subgraph is Eulerian.  In contrast, the approach
of M\"omke and Svensson is based on removing an odd-join of the graph,
which yields a possibly disconnected Eulerian subgraph.  Thus, to
maintain connectivity, one must double, rather than remove, some of
the edges in the odd-join.  The key step in proving the approximation
guarantee of the algorithm is to show that many edges will actually be
removed and relatively few edges will be doubled, resulting in a
connected, Eulerian subgraph with few edges.  First, M\"omke and
Svensson design a circulation network, $C(G,T)$, which is constructed
based on the input graph $G$, an optimal solution for a linear
programming relaxation for graph-TSP, and a carefully chosen spanning
tree $T$ (see Appendix A.1 in \cite{momke2016removing}).  Using
techniques of Naddef and Pulleyblank~\cite{naddefpulleyblank}, M\"omke
and Svensson show how to sample an odd-join of size $|E'|/3$, where
$E'$ is the edge set of a 2-vertex-connected subgraph of $G$, chosen
via the circulation network $C(G,T)$.  After removing this odd-join
from $G$, the number of doubled edges that are
added back to guarantee connectivity (while maintaining the Eulerian
property) is
directly related to the minimum cost of a circulation of $C(G,T)$.
Lemma 4.2 from \cite{momke2016removing} relates this cost to the size
of a solution output by their algorithm.
\begin{lemma} {\bf{\cite{momke2016removing}}}\label{lem:MS}
Given a 2-vertex-connected graph $G$ and a depth-first-search (DFS) tree $T$
of $G$, let $z^*$ be a circulation for $C(G,T)$ with cost $c(z^*)$.
Then there is a spanning Eulerian multigraph in $G$ with at
most $\frac{4}{3}n + \frac{2}{3} c(z^*)$ edges.
\end{lemma}

We discuss the circulation network $C(G,T)$ further in Section
\ref{sec:prelim}.  For the moment, we emphasize that if one can prove
a better upper bound on the value of $c(z^*)$, then this directly
implies an improved upper bound on the number of edges in a tour
output by M\"omke and Svensson's algorithm.

\subsection{Our contribution}

We consider the graph-TSP problem for {\em subquartic} graphs (i.e.,
graphs in which each vertex has degree at most four).  The best-known
approximation guarantee for these graphs is inherited from the general
case, even when the graph is 4-regular, and is therefore $7/5$ due to
Seb\H{o} and Vygen.  For subquartic graphs, we give an improved upper
bound on the minimum cost of a circulation for $C(G,T)$.  Using Lemma
\ref{lem:MS}, this leads to an improved approximation guarantee of
$25/18$ for graph-TSP on these graphs.  Before we give an overview of
our approach, we first explain our motivation for studying graph-TSP
on this restricted class of graphs.

As mentioned in Section \ref{sec:prog}, graph-TSP is now known to be
approximable to within $4/3$ for subcubic graphs.  So, on the one
hand, trying to prove the same guarantee for subquartic graphs is
arguably a natural next step.  Additionally, it is a well-motivated
problem to study the graph-TSP on sparse graphs, because the support
of an optimal extreme point solution for the standard linear
programming relaxation (reviewed in Section \ref{sec:LP}) has at most
$2n-1$ nonzero edges (see Theorem 4.9 in
\cite{cornuejols1985traveling}).  Thus, any graph that corresponds to
the support of such an optimal solution for the standard linear
program has average degree less than four.

However, our actual motivation for studying graphs with degree at most
four has more to do with understanding the M\"omke-Svensson algorithm
than with an abstract interest in subquartic graphs.  The basic
approach to computing an upper bound on the minimum value of $C(G,T)$
used in both \cite{momke2016removing} and \cite{mucha2014frac} is to
specify flow values on the edges of a particular circulation network
that are functions of an optimal solution for the linear programming
relaxation for graph-TSP on the graph $G$.  The cost of the
circulation obtained using these values can be analyzed in a local,
vertex by vertex manner.  Mucha showed that vertices with degree four
or five potentially increase the cost of the circulation the
most~\cite{mucha2014frac}.  In fact, one could hypothetically
construct a tight example for Mucha's analysis of the M\"omke-Svensson
algorithm on a graph where each vertex has degree at most four (or
where each vertex has degree at most five).  It therefore seems
worthwhile to determine if the cost of the circulation can be improved
on subquartic graphs.  Our results actually hold for a slightly more
general class of graphs than subquartic graphs: they hold for any
graph that has an optimal solution for the standard linear programming
relaxation of graph-TSP with subquartic support.

\subsection{Organization}

In Section \ref{sec:LP}, we discuss the standard linear programming
relaxation for graph-TSP, and in Section \ref{sec:stc}, we discuss the
circulation network $C(G,T)$ and how it can be used to find a short
tour.  In Section \ref{sec:lpn}, we show that if, for a subquartic
graph, the optimal solution for the linear program has value equal to
the number of vertices in $G$, then the network $C(G,T)$ has a
circulation with cost zero, implying that the M\"omke-Svensson algorithm
has an approximation ratio of $4/3$.  This observation provides us
with some intuition as to how one might attempt to design a better
circulation for general subquartic graphs.

In Section \ref{sec:gen}, we describe three different methods to
obtain feasible circulations.  In Section \ref{sec:x-circ}, we detail
the method used by M\"omke and Svensson and Mucha, which becomes
somewhat simpler in the special case of subquartic graphs.  This
method directly uses values from the optimal solution for the linear
program to obtain a feasible solution for the circulation network.  In Section
\ref{sec:f-circ}, we present a new method that ``rounds'' the values
from the optimal solution for the linear program.  The latter
circulation alone leads to an improved analysis over $13/9$ for
subquartic graphs, but it does not improve on the best-known guarantee
of $7/5$.  However, as we show in Section \ref{sec:combine}, if we
take the best of these two circulations, we can show that at least one
of the circulations will lead to an approximation guarantee of at most
$46/33$.  Next, in Section \ref{sec:h-circ} we consider a third method
based on extreme point structure to upper bound the optimal cost of a
circulation.  Combining all three analyses, we obtain an approximation
guarantee of $25/18$.

We remark that our notation differs from that in
\cite{momke2016removing} and \cite{mucha2014frac}, even
though we are using exactly the same circulation network and we use
their approach for obtaining the feasible circulation described in
Section \ref{sec:x-circ}.  This different notation allows us
to more easily analyze the tradeoff between the different
circulations.

\section{Preliminaries: Notation and definitions}\label{sec:prelim}

For $S \subset V$, let $\delta(S) \subset E$ denote the set of edges
with exactly one endpoint in $S$.  For $S_1, S_2 \subset V$ such that
$S_1$ and $S_2$ are disjoint, let $(S_1,S_2)$ denote the edges with
exactly one endpoint in $S_1$ and the other endpoint in $S_2$.
Throughout this paper, we make use of the following well-studied
linear programming relaxation for graph-TSP.

\subsection{Linear program for graph-TSP}\label{sec:LP}

For a graph $G=(V,E)$, the following linear program is a relaxation of
graph-TSP.  We refer the reader to Section 2 of \cite{momke2016removing}
for a discussion of its derivation and history.
\begin{align*}
& \min  \sum_{e \in E} y_e \\
y(\delta(S)) & \geq   2  ~~ {\text{for }} \emptyset \neq S \subset V,  \tag{$LP(G)$}\label{LP} \\
y & \geq  0. 
\end{align*}
We denote the feasible region of this linear program by \ref{LP}, and
we denote the value of an optimal solution by \optLP.  If $x \in$
\ref{LP} and the sum of coordinates of $x$ equals \optLP, then we say
that $x$ is an {\em optimal solution for \eqref{LP}}.  Let $n$ denote
the number of vertices in $V$.

We want to exploit certain properties of an extreme point of \eqref{LP}
such as the fact that it has sparse support.  
An extreme point $x^* \in $ \ref{LP} has at most $2n-1$ edges (see
Theorem 4.9 in \cite{cornuejols1985traveling}).  In fact, we can
obtain a more accurate bound in terms of the number of {\em heavy}
vertices.

\begin{definition}\label{def:heavy}
A vertex $v \in V$ is called {\em heavy} with respect to $y \in $
\ref{LP} if $y(\delta(v)) > 2$.  Let $H(y)$ denote the set of
heavy vertices in $V$ with respect to $y$.
\end{definition}

It is known that the number of nonzero edges in the support of an
extreme point $x^* \in $ \ref{LP} is at most the number of tight
constraints in a maximal laminar family, which can have size at most
$2n-|H(x^*)|-1$, since out of the tight constraints in a maximal laminar
family, at most $n-1$ of them can be attributed to tight sets
containing of two or more vertices~\cite{cornuejols1985traveling}.

For the sake of
simplicity, we also want to use an optimal solution $x \in $ \ref{LP}
such that $x \leq 1$. The following lemma shows that we can find such
an optimal solution (but it might not be an extreme point).

\begin{lemma}\label{lem:xisone}
Let $G=(V,E)$ be a 2-edge-connected graph.  
Then there exists $x \in$ \ref{LP}, $x \leq 1$ such
that $\sum_{e \in E} x_e =$ \optLP.
\end{lemma}

\begin{proof}
Let $x \in $ \ref{LP} be an extreme point such that $\sum_{e \in E} x_e =$
\optLP.  Suppose that there is some $x_e > 1$ for $e \in E$. First, we
observe that edge $e$ must belong to some tight cut (i.e., there
exists $S \subset V$ such that $e \in \delta(S)$ and $x(\delta(S)) =
2$).  Since $G$ is 2-edge-connected, each cut crossing edge $e$ must
include at least one other edge.  If $x_e$ did not belong to any tight
cut, then we could decrease the value of $x_e$ and obtain a smaller
solution, which is a contradiction to the optimality of $x$.

Next, we show that edge $e$ can be in at most one tight cut.
Towards a contradiction, suppose that $e$ belongs to at least two tight
cuts.  Consider the cuts $(S \cup A, V \setminus{(S \cup A)})$ and $(S \cup B,
V\setminus{(S \cup B)})$, where $S, A$ and $B$ are disjoint and
$x(\delta(S \cup A)) = x(\delta(S \cup B)) = 2$.  Suppose
that $e=ij$ and $i \in S$ and $j \in V\setminus{(S \cup A \cup B)}$.
Then edge $e$ crosses both these cuts (i.e., $e \in \delta(S \cup A)$
and $e \in \delta(S \cup B)$).  
\begin{eqnarray*}
\delta(S \cup A) & = & (S,V\setminus{(S \cup A \cup B)}) + (S,B) +  (A, B)
+ (A, V\setminus{(S \cup A \cup B)}),\\
\delta(S \cup B) & = &  (S,V\setminus{(S \cup A \cup B)}) + (S,A) + (A,B) +
(B, V\setminus{(S \cup A \cup B)}).
\end{eqnarray*}
Then we have
\begin{eqnarray}
x(\delta(S \cup A)) + x(\delta(S \cup B)) & = & 2 \cdot x(S,
V\setminus{(S \cup A
  \cup B})) + 2 \cdot x(A,B)  + x(B,S) + x(A,S) \label{eq:new}\\ && + ~ x(A, V
\setminus{(S \cup A \cup B)}) + x(B, V\setminus{(S \cup A \cup
  B)}).\nonumber
\end{eqnarray}
Since both of these cuts are tight and since $x_e > 1$ and $e \in
(S, V \setminus{(S \cup A \cup B)})$, it follows that
\begin{eqnarray*}
x(\delta(S\cup A)) + x(\delta(S\cup B)) - 2\cdot x(S,
V\setminus{(S \cup A\cup B)}) & < & 2.
\end{eqnarray*}
By \eqref{eq:new}, this implies that
\begin{eqnarray}
2\cdot x(A,B) +  x(B,S)  +  x(A,S) + x(A, V\setminus{(S \cup A \cup B)}) +
x(B, V\setminus{(S \cup A \cup B)}) ~ < ~ 2. \label{eq:new2}
\end{eqnarray}
However, we know that $x(A\cup B, ~V\setminus{(A\cup B)}) \geq 2$.
\begin{eqnarray}
x(A\cup B, V\setminus{(A \cup B)}) & = & x(A,S) + x(A,
V\setminus{(S \cup A \cup B)}) \label{eq:new3} \\ && +~ x(B,S)
+ 
x(B, V\setminus{(S \cup A \cup B)} ).\nonumber
\end{eqnarray}
Since the quantity in \eqref{eq:new3} is at most the quantity on the
left-hand side of \eqref{eq:new2}, it must be strictly less than 2,
which is a contradiction.  We can conclude that the edge
$e$ occurs in at most one tight cut.

Therefore, the only case to consider is when $e$ belongs to exactly
one tight cut.  Since there is at least one other edge (call it $f$)
besides $e$ crossing this cut (since $G$ is 2-edge-connected) and
$x_f$ must have value strictly less than 1, we can increase $x_f$ and
decrease $x_e$.  Since the cut is still tight, the solution is still
feasible and has the same value as the original solution.  Observe
that if $x_e = 1 + \eta$, we can simply decrease $x_e$ to 1 and
increase $x_f$ by $\eta$.  Indeed, if this increase of $x_e$ and
decrease of $x_f$ results in another cut, which is crossed by edge $e$,
becoming infeasible (i.e., having value strictly less than 2), then there
is some value $\eta'$ (for $0 < \eta' < \eta$) such that after
decreasing $x_e$ by $\eta'$ and increasing $x_f$ by $\eta'$, we have
two tight cuts crossed by edge $e$.  Since this cannot happen by our
previous arguments, we conclude that we can simply decrease $x_e$ to 1
and increase $x_f$ by $\eta$.

Thus, for each edge $e$ with $x_e > 1$, we can simply find a tight cut
(which is a minimum cut and we can enumerate the minimum cuts in
polynomial time), decrease $x_e$ to 1 and increase another edge so
that this cut remains tight.  Since we can do this at most $|E|$
times, we claim that this procedure can be performed efficiently.
\end{proof}

We can actually assume that $G$ is 2-vertex-connected (see Lemma 2.1
from \cite{momke2016removing}).  We can also assume that there is an
optimal extreme point $x^* \in$ \ref{LP} whose support is $E$. (If
not, we can restrict $G$ to the support of $x^*$, which does not
increase the value \optLP.  Should the resulting graph not be
2-vertex-connected, we can break it up into 2-vertex-connected
components, find an optimal extreme point solution for each component
and repeat.)  Thus, we have the following corollary of Lemma
\ref{lem:xisone}.

\begin{corollary}\label{cor:x}
Let $G=(V,E)$ be a 2-vertex-connected graph.  Suppose there exists an
extreme point $x^* \in $ \ref{LP} with support $E$ such that $\sum_{e
  \in E} x_e^* =$ \optLP.  Then there is an $x \in $ \ref{LP} such
that the following properties hold:

\begin{enumerate}[(i)]

\item $\sum_{e \in E} x_e = \sum_{e \in E} x^*_e =$ \optLP.

\item The support of $x$ and the support of $x^*$ are the same.

\item The support of $x$ and $x^*$ contains $|E| \leq 2n-|H(x^*)|-1$ edges.

\item $x \leq 1$.

\end{enumerate}
\end{corollary}

For the rest of this paper, we assume that $G=(V,E)$ is
2-vertex-connected and that $x^*$ is an optimal extreme point of
\eqref{LP} with support $E$.  Moreover, we fix $x \in$ \ref{LP} to have
the properties stated in Corollary \ref{cor:x}.  We will refer to the
set of values $\{x_e\}$ for $e \in E$ as $x$-values.  Let $\sum_{e \in
  E} x_e =$ \optLP $= (1+\eps)n$ for some $\eps$, where $0 \leq \eps
\leq 1$.

\begin{definition}
The {\em excess} $x$-value $\eps(v)$ at a vertex $v$ is the amount by
which the total value on the incident edges exceeds 2 (i.e., $\eps(v) =
x(\delta(v)) -2$).
\end{definition}

The following fact will be useful in our analysis.
If \optLP $= (1+\eps)n$, then
\begin{eqnarray*}
\sum_{v \in V} \x(\delta(v)) & = & \sum_{v \in V} (2 + \eps(v)) ~~ =
~~ 2 (1 + \eps)n.
\end{eqnarray*}
This implies, 
\begin{eqnarray}
\sum_{v \in V} \eps(v) & = & 2 \eps n.\label{eps-vertices}
\end{eqnarray}

\subsection{Spanning trees and circulations}\label{sec:stc}

Let us recall some useful definitions from the approach of M\"omke and
Svensson~\cite{momke2016removing} that we use throughout this
paper.

\begin{definition}\label{def:greedy-tree}
Let $y \in $ \ref{LP}.  A {\em greedy DFS tree} chosen with respect to $y$ is
a spanning tree formed via a depth-first search of $G$.  If there is a
choice as to which edge to traverse next, the edge with the highest
$y$-value is chosen.
\end{definition}

For a given graph $G$ and a solution $y \in $ \ref{LP}, let $T$ denote
a greedy DFS tree with respect to $y$.  Suppose $T$ has root $r$, and
let $E(T)$ denote the edges in $T$ (i.e., tree edges).  We orient
$E(T)$ to be an arborescence with root $r$, and we orient $B(T):= E
\setminus{E(T)}$ ``backwards," that is, so that each edge in $B(T)$
forms a directed cycle with a path of the tree.  This is possible
since $T$ is a DFS tree.  We use the notation $(i,j)$ to denote an
edge directed from $i$ to $j$.  Note that once we have fixed a tree
$T$, all edges in $E$ can be viewed as directed edges.  When we wish
to refer to an undirected edge in $E$, we use the notation $ij \in E$.
With respect to the greedy DFS tree $T$, we have the following
definitions.

\begin{definition}\label{def:internal}
An {\em internal} node in $T$ is a vertex that is neither the root of
$T$ nor a leaf in $T$.  We use $T_{int}$ to denote this subset of
vertices.
\end{definition}

\begin{definition}\label{def:branch}
A {\em branch} vertex in $T$ is a vertex with at least two outgoing
tree edges.
\end{definition}

Note that the root of $T$ cannot be a branch vertex since $G$ is
2-vertex-connected.

\begin{definition}\label{def:expensive}
An {\em expensive} vertex is a vertex in $T_{int}$ with two incoming
edges that belong to $B(T)$.  We use $T_{exp}$ to denote this subset
of vertices.
\end{definition}

As we will see in Lemma \ref{lem:four}, in a subquartic graph,
expensive vertices are the only vertices (besides the root) that can contribute to the
cost of $C(G,T)$ (which is of interest due to Lemma \ref{lem:MS}).
For the sake of simplicity, we sometimes ignore the contribution of
the root in our calculations, since the contribution of the root to
the cost of $C(G,T)$ is negligible (at most 2).

\begin{fact}\label{exp:bound} 
The number of expensive vertices is
bounded as follows: $|T_{exp}| \leq |B(T)|/2 \leq n/2$.
\end{fact}

\begin{lemma}\label{lem:bex}
If $G$ is subquartic, then a branch vertex in $T_{int}$ is not expensive.
\end{lemma}

\begin{proof}
In a graph with vertex degree at most four, a branch vertex can have
at most one incoming back edge and therefore cannot be expensive.
\end{proof}

\begin{definition}
A {\em tree cut} is the partition of the vertices of the tree $T$
induced when we remove an edge $(u,v) \in T$.
\end{definition}

For each edge $(i,j) \in B(T)$, let $b(i,j) \leq 1$ be a nonnegative value.

\begin{definition}
Consider a tree cut corresponding to edge $(u,v) \in T$, and remove all
back edges $(w,u) \in B(T)$, where $w$ belongs to the subtree of $v$
in $T$.  We say that the remaining back edges that cross this tree cut
{\em cover} the cut.  If the total $b$-value of the edges that cover
the cut is at least 1, then we say that this tree cut is {\em
  satisfied by $b$}.
\end{definition}

We extend this definition to the vertices of $T$.
\begin{definition}
A vertex $v$ in $T$ is {\em satisfied by $b$} if for each incident
outgoing edge in $T$, the corresponding tree cut is satisfied by $b$.
On the other hand, if there is at least one incident outgoing edge
whose corresponding tree cut is not satisfied by $b$, then the vertex
$v$ is {\em unsatisfied by $b$}.
\end{definition}

M\"omke and Svensson define a circulation network,
$C(G,T)$ (see Appendix A.1 of \cite{momke2016removing}), and use the
cost of a feasible circulation to upper bound the length of a TSP tour
in $G$.  (See Lemma \ref{lem:MS}.)\footnote{In the journal version,
 M\"omke and Svensson give a linear program
  $LP(G,T)$ (\cite{momke2016removing}, page 11) whose objective value
  is equal to the minimum cost of a circulation in $C(G,T)$.  We
  choose to use the notation $C(G,T)$ to avoid confusion between this
  optimization problem and the linear program \eqref{LP} defined in
  Section \ref{sec:LP}.}

\begin{lemma}\label{lem:four}
Let $G$ be a subquartic graph and $T$ be a DFS tree.  Let $b:
B(T) \rightarrow [0,1]$.  If each internal vertex in $T$ is satisfied
by $b$, then there is a feasible integer circulation of $C(G,T)$ whose
cost (not including the contribution of the root) is
upper bounded by the following function:
\begin{eqnarray}
\sum_{j \in T_{exp}} \max \left\{ 0, ~\left(\sum_{i:(i,j) \in B(T)}
b(i,j)\right) -1\right\}.\label{cost-func-main}
\end{eqnarray}
\end{lemma}

\begin{proof} 
The purpose of the circulation network $C(G,T)$ in
\cite{momke2016removing} is to find a subset of back edges, $B'
\subseteq B(T)$, so that $T\cup B'$ is 2-vertex-connected and so that
the following cost function 
(stated here for a subquartic graph) 
is minimized:
\begin{eqnarray}
\sum_{j \in T_{exp}} \max \left\{ 0, ~\left(\sum_{i:(i,j) \in B'} 1
\right) -1\right\}. \label{int-cost-func}
\end{eqnarray}

If $b': B(T) \rightarrow \{0,1\}$ is an integral function and every
internal vertex in $T$ is satisfied by $b'$, then the edges with
$b'$-value 1 form a set $B'$ such that $T\cup B'$ is
2-vertex-connected.  To obtain an upper bound on the cost of such a
2-vertex-connected subgraph, we can use any fractional values $b:B(T)
\rightarrow [0,1]$ such that each internal vertex in $T$ is satisfied
by $b$.  This is made explicit in the formulation of $LP(G,T)$ and the
subsequent discussion on page 11 of \cite{momke2016removing}; observe
that these fractional $b$-values are a feasible solution for
$LP(G,T)$.  In the case of subquartic graphs, the only vertices that
can contribute to the objective function in $LP(G,T)$ (and therefore
to the cost function \eqref{cost-func-main}) are the expensive
vertices, since the maximum value allowed on an edge is 1.  Thus, in
this special case, the simplified cost function \eqref{cost-func-main}
is equivalent to the objective function for $LP(G,T)$ and hence for
the cost of the circulation for
the network $C(G,T)$ used in \cite{momke2016removing}.
For a formal description of $C(G,T)$ and for the equivalence between
the objective function of $LP(G,T)$ and the cost of a circulation of
$C(G,T)$, we refer the reader to the proof of Lemma 4.1 in
\cite{momke2016removing}.
\end{proof}

From Lemma \ref{lem:MS}, we see that if we find a circulation with cost
zero for $C(G,T)$, then $G$ has a TSP tour of length at most
$\frac{4}{3}n$.  We note that in \cite{mucha2014frac}, it is shown
that for a graph $G$ for which \optLP $= n$, there is a circulation with
cost at most $\frac{n}{6}$ for $C(G,T)$, which results in a TSP tour
of length at most $\frac{13}{9}n$.

\subsection{Circulations from feasible LP solutions}

How do we choose a function $b : B(T) \rightarrow [0,1]$ so that every
internal vertex is satisfied by $b$?  One natural approach is to begin
with a feasible solution $y \in $ \ref{LP}.  The following useful
lemma shows that certain vertices will be satisfied by $y$.

\begin{lemma}\label{lem:tight-vertex}
For $y \in $ \ref{LP}, $ y \leq 1$, let $T$ be a greedy DFS tree
chosen with respect to $y$, and let $y(i,j) = y_{ij}$ for all back
edges in $B(T)$.  Then for any $v \in T_{int}$, $v$ is satisfied by
$y$ if $y(\delta(v)) = 2$ or if $v$ is a nonbranch vertex with no
outgoing back edge (e.g., $v$ is an expensive vertex if $G$ is subquartic).
\end{lemma}

\begin{proof}
First, we prove that $v$ is satisfied by $y$ if $y(\delta(v)) = 2$.
Let $(v,t)$ be an outgoing tree edge from vertex $v$.  Partition the
edges of $\delta(v)$ into the following two sets: the first set $E_1$
consists of edge $(v,t)$ and the incoming back edges from the subtree
rooted at $t$.  Let $E_2$ denote the remaining edges in $\delta(v)$.
Observe that at least one of the sets $E_1$ and $E_2$ has total
$y$-value at most 1, since these two sets form a partition of
$\delta(v)$.  Let $B' \subset B(T)$ denote the back edges that cover
the tree cut corresponding to $(v,t)$.  Observe that both $B' \cup
E_1$ and $B' \cup E_2$ are edge sets that cross cuts in $G$ (which we
assume to be the support of $y$).  In other words, let $V_t$ denote
the vertices in the subtree of $T$ rooted at $t$ (including vertex
$t$).  Then $B' \cup E_1 = \delta(V_t)$ and $B' \cup E_2 = \delta(V_t
\cup v)$.  Since either $E_1$ or $E_2$ has total $y$-value at most 1,
we can conclude that $B'$ has $y$-value at least 1, and therefore the
tree cut corresponding to $(v,t)$ is satisfied by $y$.  By repeating
this argument for each outgoing tree edge (in the case that $v$ is a
branch vertex), we can conclude that $v$ is satisfied by $y$.

Now let us consider the case in which $v$ is a nonbranch vertex with
no outgoing back edge (i.e., all back edges are incoming).  Then the
edge set $E_2$ consists of a single edge and we can conclude that the
$y$-value of the edges in $B'$ is at least 1.
\end{proof}

The following lemma pertains to vertices not satisfied by a feasible
solution $y \in $ \ref{LP}.

\begin{lemma}\label{lem:tree-cut-sat}
For a subquartic graph $G$, and $y \in $ \ref{LP}, $ y \leq 1$, let
$T$ be a greedy DFS tree chosen with respect to $y$ and let $y(i,j) =
y_{ij}$ for all back edges in $B(T)$.  Then for any $u \in T_{int}$,
$u$ has at most one outgoing tree edge whose corresponding tree cut is
not satisfied by $y$.
\end{lemma}

\begin{proof} 
Let $(u,v)$ be a tree edge such that there is no back edge coming into
vertex $u$ from the subtree of $T$ rooted at $v$.  Then observe that
since the $y$-value of edge $(u,v)$ is at most 1, the total $y$-value
of the edges that cover this tree cut must be at least $1$, since the
covering edges plus the edge $(u,v)$ have $y$-value at least 2 (by the
constraints in \eqref{LP}).  Thus, the tree cut corresponding to edge
$(u,v)$ must be satisfied by $y$.

If $u$ has only one outgoing tree edge, then the lemma holds.  If $u$
is a branch vertex with two outgoing tree edges, there is at most one
incoming back edge to vertex $u$.  Thus, the tree cut corresponding to
at least one outgoing tree edge, say $(u,v)$, is satisfied by $y$,
since $u$ has no incoming back edge from the subtree rooted at $v$.
Moreover, if $u$ has three outgoing tree edges, then there is no
incoming back edge to vertex $u$ and therefore the tree cuts
corresponding to each of the three outgoing tree edges are satisfied
by $y$.
\end{proof}

\section{Subquartic graphs: \optLP $= n$}\label{sec:lpn}

We now show that in the special case when \optLP $=n$ and $G$ is
subquartic, there is a circulation with cost zero.  Recall the
$\{x_{ij}\}$ values from Corollary \ref{cor:x}, and assume that
$\sum_{ij \in E} x_{ij} = n$.  Note that if $|E| = n$, then each edge
in $E$ must have $x$-value 1.  Thus, $G$ is a Hamilton cycle.  If $|E|
> n$, then we can choose a greedy DFS tree $T$ with respect to $x$
such that each edge $ij \in E$ with $x$-value $x_{ij} = 1$ (a
``$1$-edge'') belongs to $T$.

\begin{lemma}\label{lem:one-edge}
If \optLP $=n$ and $|E| > n$, then there is a greedy DFS tree chosen
with respect to $x$ such that all $1$-edges belong to $E(T)$.
\end{lemma}

\begin{proof} Observe that 
a vertex with degree at least three can have at most one incident
$1$-edge, since the $x$-value at each vertex is exactly 2 when \optLP
$=n$ (i.e., $x(\delta(v)) = 2$ for all $v \in V$).  By the assumptions
in the lemma, there is some vertex, say $i$, with degree at least
three, which therefore has at most one incident $1$-edge.  Thus, we
choose $i$ to be the root of the greedy DFS tree.  If $i$ is incident
to a $1$-edge, then this $1$-edge belongs to the resulting tree by the
rules defining the construction of a greedy DFS tree.

Suppose that after we are done constructing the greedy DFS tree, there
is a back edge $(u,v)$ that has $x$-value 1.  Then when vertex $v$ was
visited in the depth-first search, it should have traversed this edge
as the next tree edge.  Otherwise, the edge it did traverse/add to the
tree also had an $x$-value of 1, which is a contradiction because as
an internal vertex with an incoming back edge, vertex $v$ has degree
at least three.\end{proof}

For the rest of Section \ref{sec:lpn}, let $T$ denote a greedy DFS
tree in which all $1$-edges are tree edges.

\begin{lemma}\label{lem:two-edge}
When \optLP $=n$, and each back edge $(i,j) \in B(T)$ is assigned
value $f(i,j) = 1/2$, then each vertex in $T_{int}$ is satisfied by
$f$.
\end{lemma}

\begin{proof}
Set $x(i,j) = x_{ij}$ for each back edge.  Then each tree cut is
satisfied by $x$, because each vertex $v \in V$ has $x(\delta(v))=2$,
and we can therefore apply Lemma \ref{lem:tight-vertex}.  Since there
are no $1$-edges in the set of back edges, this implies that each tree
cut must in fact be covered by at least two edges.  Thus, setting
$f(i,j) = 1/2$ results in each tree cut being satisfied by $f$.
\end{proof}

We remark that Lemmas \ref{lem:one-edge} and \ref{lem:two-edge} hold
for general graphs.

\begin{lemma}\label{lem:three-edge}
If $G$ is subquartic and \optLP$=n$, setting $f(i,j) = 1/2$ for each
edge $(i,j) \in B(T)$ yields a circulation with cost zero.
\end{lemma}

\begin{proof}
This follows from the fact that each vertex in $T_{int}$ has in-degree
at most two and therefore the total $f$-value coming into a vertex is
at most one.  Thus, the circulation value is zero.  (Note that the
root can contribute $1/2$ to the circulation, but there exists a
minimum cost circulation that is integral, and its cost will therefore
still be zero.)
\end{proof}

\begin{theorem}
If $G$ is subquartic and \optLP $=n$, then $G$ has a TSP tour of
length at most $4n/3$.
\end{theorem}

\section{Subquartic graphs: General case}\label{sec:gen}

In this section, we consider the general case of subquartic graphs.
For a subquartic graph $G=(V,E)$, suppose \optLP $= (1+\eps)n$ for
some $\eps > 0$.  Recall the $\{x_{ij}\}$ values from Corollary
\ref{cor:x}.  Let $T$ be a greedy DFS tree chosen with respect to $x$
(Definition \ref{def:greedy-tree}), and let $x(i,j) = x_{ij}$ for all
back edges in $B(T)$.  As shown in Lemma \ref{lem:four}, the only
vertices that can add to the cost function are the expensive vertices
(Definition \ref{def:expensive}).  We find the following terminology
convenient.

\begin{definition}
A vertex $v \in T_{int}$ that is satisfied by $x$ is called {\em
  LP-satisfied}.
\end{definition}

\begin{definition}
A vertex $v \in T_{int}$ that is not satisfied by $x$ is called {\em
  LP-unsatisfied}.
\end{definition}

The following corollaries follow from Lemma \ref{lem:tight-vertex}.

\begin{corollary}\label{lem:expensive}
An expensive vertex is LP-satisfied.
\end{corollary}

\begin{corollary}\label{lem:heavy}
An LP-unsatisfied vertex is heavy (i.e., it belongs to $H(x)$).
\end{corollary}

The reason we emphasize that an LP-unsatisfied vertex is heavy is that
we can use the excess $x$-value of this vertex to pay for the
increased value on an edge that covers the unsatisfied tree cut
corresponding to one of its incident outgoing edges so that this tree
cut becomes satisfied.  We also wish to use the excess $x$-value of an
expensive vertex to pay for some of its contribution to the cost
function incurred by the back edges coming into the vertex.  For each
vertex $v$, we want to use the quantity $\eps(v)$ at most once in this
payment scheme.  This will be guaranteed by the fact that
LP-unsatisfied vertices and expensive vertices are disjoint sets
(Corollary \ref{lem:expensive}).

\subsection{The $x$-circulation}\label{sec:x-circ}

In this section, we use the $x$-values to obtain an upper bound on the
cost of a circulation, essentially following the arguments of M\"omke
and Svensson~\cite{momke2016removing} and Mucha~\cite{mucha2014frac}.
We present the analysis here, since we refer to it in Section
\ref{sec:combine} when we analyze the cost of taking the best of two
circulations. Also, the arguments can be somewhat simplified due to
the subquartic structure of the graph, which is useful for our
analysis.

Recall that for each back edge $(i,j)$ in $B(T)$, we have $x(i,j) = x_{ij}$.
(For a vertex $j \in T_{int}\setminus{T_{exp}}$, we can actually set
$x(i,j) = 1$, since there is at most one incoming back edge to vertex
$j$ and this does not change the worst-case analysis.)

\begin{definition}
For each vertex $j \in T_{exp}$, let $\xij \leq \xaj$ denote the
$x$-values of the two incoming back edges to vertex $j$.  Let $c_x(j)
= \xij + \xaj -1 - \eps(j)$.
\end{definition}

We will show that there is a function $x': B(T) \rightarrow [0,1]$
such that each vertex in $T_{int}$ is satisfied by $x'$ and the cost
of the circulation can be bounded by
\begin{eqnarray}
\sum_{j \in T_{exp}} \max \left\{0, \left(\sum_{i: (i,j) \in B(T)}
x'(i,j)\right) - 1 \right\} & \leq & \sum_{j \in T_{exp}}
\max\{0,c_x(j)\} + \sum_{j \in T_{int}}\eps(j). ~\label{cost-func}
\end{eqnarray}

\begin{claim}\label{clm:fact}
For an expensive vertex $j \in T_{exp}$, the following holds:
\begin{eqnarray*}
2\cdot x_{max}(j) + \xi(j) \leq 2 + \eps(j).
\end{eqnarray*}
\end{claim}

\begin{cproof}
By the construction of $T$, we note that the $x$-value of the tree
edge leaving vertex $j$ must be at least $\xaj$.  Thus, the above
inequality holds.
\end{cproof}

\begin{claim}\label{clm:thirteen}
The value $c_x(j)$ can be upper bounded as follows:
\begin{eqnarray*}
c_x(j) & \leq & \frac{\xij}{2} - \frac{\eps(j)}{2} ~ \leq ~ 1 - \xij.
\end{eqnarray*}
\end{claim}

\begin{cproof}
For a vertex $j \in T_{exp}$, we can use Claim \ref{clm:fact} to show
\begin{eqnarray*}
c_x(j) ~= ~\xa(j) + \xi(j) - 1 - \eps(j) & \leq & (2 + \eps(j) -
\xi(j))/2 + \xi(j) -1 -\eps(j)\\
& = & \frac{\xi(j)}{2} - \frac{\eps(j)}{2}.
\end{eqnarray*}
Claim \ref{clm:fact} also implies
\begin{eqnarray*}
\frac{x_{min}(j)}{2} - \frac{\eps(j)}{2} & \leq & 1-\xaj ~ \leq ~ 1-\xij.
\end{eqnarray*}
\end{cproof}

\begin{lemma}\label{clm:fourteen}
For a vertex $j \in T_{exp}$, $c_x(j) \leq 1/3$.
\end{lemma}

\begin{proof}
By Claim \ref{clm:thirteen}, we have
\begin{eqnarray*}
c_x(j) & \leq & \min\left\{\frac{\xij}{2}, ~1-\xij\right\}.
\end{eqnarray*}
This implies that $c_x(j) \leq 1/3$, which occurs when $\xij = 2/3$,
as shown by Mucha~\cite{mucha2014frac}.
\end{proof}

To make the circulation feasible, we need to increase the $x$-values
of some of the back edges in $B(T)$ so that all of the LP-unsatisfied
vertices become satisfied.  By Corollary \ref{lem:heavy}, these
vertices are heavy.  Thus, we will use the ``extra'' $\eps(v)$ for an
LP-unsatisfied vertex $v$ to ``pay" for increasing the $x$-value on an
appropriate back edge.  For ease of notation, we now set $x'(i,j) :=
x(i,j)$ for all $(i,j) \in B(T)$.  We will update these $x'$-values so
that each LP-unsatisfied vertex is satisfied by $x'$.

Consider an LP-unsatisfied vertex $v \in T$ and let $(v,t)$ denote the
outgoing tree edge corresponding to the unsatisfied tree cut.  Let $B'
\subseteq B(T)$ be the set of back edges that cover the tree cut
corresponding to edge $(v,t)$.  Partition the edges in $\delta(v)$
into the following two sets: the first set $E_1$ consists of edge
$(v,t)$ and the incoming back edges from the subtree rooted at $t$.
Let $E_2$ denote the remaining edges in $\delta(v)$.  Note that
$x(E_1) + x(E_2) = 2 + \eps(v)$ and since $v$ is LP-unsatisfied,
$\eps(v) > 0$.  From this and the following facts
\begin{eqnarray*}
x(E_1) + x(B') & \geq & 2,\\
x(E_2) + x(B') & \geq & 2,
\end{eqnarray*}
we can conclude that
\begin{eqnarray*}
x(B') & \geq & 1 - \frac{\eps(v)}{2}.
\end{eqnarray*}

Let $(i,j) \in B'$ be an arbitrary edge in $B'$, which exists because
$G$ is $2$-vertex-connected.  We will update the value of $x'(i,j)$ as
follows: $$ x'(i,j) := \min\{1, ~x'(i,j) + \eps(v)/2\}.$$ We use this
recursive notation because a back edge's value can be increased
multiple times in the process of satisfying all LP-unsatisfied
vertices.  The following lemma follows by the construction of the
$x'$-values and by Corollary \ref{lem:expensive}.

\begin{lemma}\label{lem:fix-x}
The cost of satisfying all of the LP-unsatisfied vertices is at most
$\sum_{j \in T_{int}\setminus{T_{exp}}} \eps(j)/2$.  In other words,
\begin{eqnarray*}
\sum_{(u,v) \in B(T)} (x'(u,v) - x(u,v)) & \leq & \sum_{j \in
  T_{int}\setminus{T_{exp}}} \frac{\eps(j)}{2}.
\end{eqnarray*}
\end{lemma}

Since all vertices in $T$ are now satisfied by $x'$, the $x'$-values
can be used to compute an upper bound on the cost of a feasible
circulation of $C(G,T)$.

\begin{lemma}\label{lem:two}
The function $x': B(T) \rightarrow [0,1]$ corresponds to a feasible
circulation of $C(G,T)$ with cost (not including the contribution of
the root) at most
\begin{eqnarray*}
\sum_{j \in T_{exp}} \max\{0,~c_x(j)\} + \sum_{j \in T_{int}}\eps(j).
\end{eqnarray*}
\end{lemma}

\begin{proof}
By construction, every vertex in $T_{int}$ is satisfied by $x'$.
Thus, the $x'$-values correspond to a feasible circulation of
$C(G,T)$.  The cost of the circulation based on the $x'$-values is
\begin{eqnarray*}
\sum_{j \in T_{exp}} \max \left\{0, \left(\sum_{i: (i,j) \in B(T)}
x'(i,j)\right) - 1 \right\} & \leq & \sum_{j \in T_{exp}} \max
\left\{0, \left(\sum_{i: (i,j) \in B(T)} x(i,j)\right) - 1 \right\}
\\ 
&& + \sum_{(u,v) \in B(T)} (x'(u,v) - x(u,v)).
\end{eqnarray*}
We have
\begin{eqnarray*}
\sum_{j \in T_{exp}} \max \left\{0, \left(\sum_{i: (i,j) \in B(T)}
x(i,j)\right) - 1 \right\} & = & \sum_{j \in T_{exp}} \max \left\{0,
\xaj + \xij - 1 \right\}\\
& \leq & \sum_{j \in T_{exp}} \left( \max \left\{0, \xaj + \xij - 1 -
\eps(j) \right\} + \eps(j)\right) \\
& \leq & \sum_{j \in T_{exp}} \max\{0,~c_x(j)\} + \sum_{j \in
  T_{exp}}\eps(j).
\end{eqnarray*}
Combining the above inequality with Lemma \ref{lem:fix-x} proves the
lemma.\end{proof}

\begin{theorem}\label{thm:three}
When $G$ is subquartic and \optLP $= (1+\eps)n$, there is a
feasible circulation for $C(G,T)$ with cost at most $n/6 + 2\eps n + 2$.
\end{theorem}

\begin{proof}
The number of expensive vertices is at most $n/2$ (Fact
\ref{exp:bound}) and each expensive vertex $j \in T_{exp}$ can add at
most $1/3 + \eps(j)$ to the cost function (Lemma \ref{clm:fourteen}).
Each vertex $j \in T_{int}\setminus{T_{exp}}$ can add at most
$\eps(j)$ to the cost function.  Using the fact that $\sum_{j \in
  T_{int}} \eps(j) \leq 2 \eps n$ and the fact that the contribution
of the root is at most 2 yields the theorem.
\end{proof}

\subsection{The $f$-circulation}\label{sec:f-circ}

Now we describe a new method to obtain a feasible circulation: We show
how to obtain values $f'(i,j)$ for each edge $(i,j) \in B(T)$ such
that each vertex in $T_{int}$ is satisfied by $f'$.  The values will
be used to demonstrate an improved upper bound on the cost of a
circulation of $C(G,T)$ when $G$ is a subquartic graph.  In this
section, we will prove the following theorem, which implies that the
M\"omke-Svensson algorithm has an approximation guarantee of $17/12$
for graph-TSP on subquartic graphs.
\begin{theorem}\label{thm:f}
When $G$ is subquartic and \optLP $= (1+\eps)n$, there is a
feasible circulation for $C(G,T)$ with cost at most $n/8 + 2\eps n + 2$.
\end{theorem}
Consider a vertex $j \in T_{exp}$.  If both incoming back edges have
$f$-value $1/2$, then this vertex will not contribute anything to the
cost of the circulation.  Thus, on a high level, our goal is to find
$f$-values that are as close to $1/2$ as possible, while at the same
time not creating any additional unsatisfied vertices.  For example,
in Lemma \ref{clm:fourteen}, we saw that a vertex can contribute as
much as $1/3$ to the cost of the circulation, when $x_{min} = 2/3$.
If we could decrease $x_{max}$ and $x_{min}$, we could decrease the
cost of the $x$-circulation.  This might, in turn, cause some
LP-satisfied vertices to become unsatisfied.  We can avoid this
situation by strategically increasing some other $x$-values.  The
$f$-value therefore corresponds to a decreased $x$-value if the
$x$-value is high, and an increased $x$-value if the $x$-value is low.

The remaining issue that needs to be addressed is that the $f$-values
corresponding to decreased $x$-values may cause some of the
LP-unsatisfied vertices to become more unsatisfied than they were by
the $x$-values.  However, note that in Section \ref{sec:x-circ}, we
only used $\eps(j)/2$ to satisfy an LP-unsatisfied vertex $j$.  We can
actually use up to $\eps(j)$.  This observation allows us to decrease
the $x$-values while still ensuring that all vertices are satisfied.
We use the rules depicted in Figure \ref{fig:rules} to determine the
values $f: B(T) \rightarrow [0,1]$.
\begin{figure}[h!]
\begin{center}
\fbox{
\begin{tabular}{ccc}
$x_{ij} > 3/4$ & $\Rightarrow$ & $f(i,j) = 2 x_{ij} -1,$\\
$x_{ij} < 1/4$ & $\Rightarrow$ & $f(i,j) = 2 x_{ij},$\\
$1/4 \leq x_{ij} \leq 3/4$ & $\Rightarrow$ & $f(i,j) = 1/2.$
\end{tabular}}
\end{center}
\caption{Rules for constructing the $f$-values from the
  $x$-values.}\label{fig:rules}
\end{figure}

\begin{lemma}\label{lem:sat-by-f}
If a vertex $v$ is LP-satisfied, then it is satisfied by $f$.
\end{lemma}

\begin{proof}
Let $B' \subseteq B(T)$ denote the set of back edges that covers a
particular tree cut.  If $B'$ consists of a single edge with $x$-value
1, then the $f$-value of this edge will also be 1.  Let us now suppose
the set $B'$ contains multiple edges, whose total $x$-value is at least
1.  Consider the following three cases: First, suppose $B'$ contains at
least two edges with $x$-value at least $1/2$.  In this case, the
$f$-value on each of these edges remains at least $1/2$.  Second, if
the set $B'$ contains only edges that have $x$-value at most $1/2$,
then the total $f$-value is at least the total $x$-value, since the
$f$-value does not decrease in this case.

The third case is when $B'$ contains only one edge with $x$-value at
least $1/2$.  Suppose that this edge $e$ has value $x_e=1-\gamma \geq
1/2$.  The remaining edges in $B'$ must have total $x$-value at least
$\gamma$.  If at least one of these edges' $x$-value is at least $1/4$,
then we are done (because this edge will have $f$-value $1/2$).  Thus,
all the edges in the set $B'\setminus{e}$ must have $x$-value less
than $1/4$.  In this case, the total $f$-value for these edges is at
least $2\gamma$.  Note that the $f$-value of edge $e$ is at least
$1-2\gamma$.
\end{proof}

\begin{definition}
For each vertex $j \in T_{exp}$, let $c_f(j) = \sum_{i:(i,j) \in B(T)}
f(i,j) - 1 - \eps(j)$.
\end{definition}

For ease of notation, set $f'(i,j) := f(i,j)$ for all $(i,j) \in B(T)$.

\begin{lemma}\label{lem:sat-v}
For an LP-unsatisfied vertex $v \in T_{int}$, if we increase by the
amount $\eps(v)$ the $f'$-value of an edge that covers its unsatisfied
tree cut, then vertex $v$ will be satisfied by $f'$.
\end{lemma}

\begin{proof}
We will argue, as we did in Section \ref{sec:x-circ}, that each
LP-unsatisfied vertex $v \in T_{int}\setminus{T_{exp}}$ can be
satisfied by increasing the $f'$-value of a single back edge that
covers the unsatisfied tree cut corresponding to one of its outgoing
tree edges, $(v,t)$.

Let $B' \subseteq B(T)$ be the set of back edges that cover the tree
cut corresponding to edge $(v,t)$.  Partition the edges in $\delta(v)$
into the following two sets: the first set $E_1$ consists of edge
$(v,t)$ and the incoming back edges from the subtree rooted at $t$.
Let $E_2$ denote the remaining edges in $\delta(v)$.  Then we have
\begin{eqnarray*}
x(E_1) + x(B') & \geq & 2,\\
x(E_2) + x(B') & \geq & 2.
\end{eqnarray*}
Since $v$ is LP-unsatisfied, $x(B') = 1-\gamma < 1$.  Thus,
\begin{eqnarray*}
x(E_1) + x(E_2) & \geq & 2 + 2\gamma.
\end{eqnarray*}
So $\eps(v) \geq 2\gamma$.  Therefore, since the $f$-value of the edges
in $B'$ is at least $1-2\gamma$, the amount $2\gamma$ is sufficient to
``correct'' the $f$-values so that $v$ is satisfied by
$f'$.
\end{proof}

The following lemma follows by the construction of the $f'$-values and
by Corollary \ref{lem:expensive}.

\begin{lemma}\label{lem:fix-f}
The cost of satisfying all of the LP-unsatisfied vertices is at most
$\sum_{j \in T_{int}\setminus{T_{exp}}} \eps(j)$.  In other words,
\begin{eqnarray*}
\sum_{(u,v) \in B(T)} (f'(u,v) - f(u,v)) & \leq & \sum_{j \in
  T_{int}\setminus{T_{exp}}} \eps(j).
\end{eqnarray*}
\end{lemma}

\begin{lemma}\label{lem:four-five}
For a vertex $j \in T_{exp}$, $c_f(j)  \leq  \frac{1}{4}.$
\end{lemma}

\begin{proof}
We break the proof into the following claims.
\begin{claim}\label{clm:end1}
For $j \in T_{exp}$, if $\xij \geq 1/2$ or if $\xaj \leq 3/4$, then
$c_f(j) \leq 0$.
\end{claim}

\begin{cproof}
We consider the following three cases.

\vspace{2mm}
\noindent
{\bf Case (i)} First, we consider the case in which $\xij \geq 3/4$.
Then, the total $f$-value of the back edges coming into vertex $j$ is
\begin{eqnarray*}
c_f(j) + 1 + \eps(j)  & = & 2 \cdot \xaj -1 + 2 \cdot \xij -
1\label{incoming} \\
& \leq & (2 + \eps(j) - \xij) + 2 \cdot \xij -2\\
& = & \xij + \eps(j).
\end{eqnarray*}
The inequality follows from Claim \ref{clm:fact}.
This implies that
\begin{eqnarray*}
c_f(j) & \leq &  \xij - 1 ~ \leq ~ 0.
\end{eqnarray*}

\vspace{2mm}
\noindent
{\bf Case (ii)} Now let us consider the case when $\xaj \geq 3/4$ and
$1/2 \leq \xij \leq 3/4$.  The total $f$-value of the incoming back
edges is
\begin{eqnarray*}
c_f(j) + 1 + \eps(j) ~ = ~ 2\cdot \xaj - 1 + \frac{1}{2} & = & 2 \cdot
\xaj - \frac{1}{2} \\ 
& \leq & 3/2 - \xij + \eps(j).
\end{eqnarray*}
The inequality follows from Claim \ref{clm:fact}.
This implies that
\begin{eqnarray*}
c_f(j) & \leq & 1/2 - \xij.
\end{eqnarray*}
Since $\xij \geq 1/2$, this implies that $c_f(j) \leq 0$.

\vspace{2mm}
\noindent
{\bf Case (iii)} Now let us consider the case when $\xaj \leq
3/4$.  Note that in this case, the $f$-value for each incoming back
edge is at most $1/2$.  Thus, $c_f(j) \leq 0$.
\end{cproof}

It remains to examine the case when $\xaj > 3/4$ and $0 < \xij \leq
1/2$.  This is the only situation when $c_f(j)$ can be positive.

\begin{claim}\label{clm:end2}
If $\xaj \geq 3/4$ and $0 < \xij \leq 1/2$, then $c_f(j) \leq
\min\{\xij, ~1/2-\xij \}$.
\end{claim}

\begin{cproof}
\vspace{2mm}
\noindent
{\bf Case (iv)} Now let us consider the case when $\xaj \geq 3/4$ and
$1/4 \leq \xij < 1/2$.  Applying Claim \ref{clm:fact}, we see that 
the total $f$-value of the incoming back edges
is
\begin{eqnarray*}
c_f(j) + 1 + \eps(j) ~ = ~ 2\cdot \xaj  - 1 + \frac{1}{2} & \leq &
\frac{3}{2} - \xij + \eps(j).
\end{eqnarray*}
Therefore,
\begin{eqnarray*}
c_f(j) & \leq & \frac{1}{2} - \xij ~ \leq ~ \xij.
\end{eqnarray*}

\vspace{2mm}
\noindent
{\bf Case (v)} Now let us consider the case when $\xaj \geq 3/4$ and
$0 < \xij < 1/4$.  Applying Claim \ref{clm:fact}, the
total $f$-value of the incoming back edges is
\begin{eqnarray*}
c_f(j) + 1 + \eps(j) ~ = ~ 2\cdot \xaj  - 1 + 2 \cdot \xij  & \leq &
\xij + 1 + \eps(j).
\end{eqnarray*}
Therefore,
\begin{eqnarray*}
c_f(j) & \leq & \xij ~ \leq ~ \frac{1}{2} - \xij.
\end{eqnarray*}
\end{cproof}
Claims \ref{clm:end1} and \ref{clm:end2} show that $c_f(j) \leq
\min\{x_{min}(j), 1/2 - x_{min}(j)\}$ for $0 \leq x_{min}(j) \leq 1/2$
and $c_f(j) \leq 0$ otherwise.  Thus, $c_f(j) \leq 1/4$.\end{proof}

We can now prove Theorem \ref{thm:f}.

\noindent
{{\it{Proof of Theorem \ref{thm:f}.}}}
We show that the function $f': B(T) \rightarrow [0,1]$ corresponds to a feasible
circulation of $C(G,T)$ with cost (not including the root, which adds
2) at most
\begin{eqnarray}
\sum_{j \in T_{exp}} \max\{0,~c_f(j)\} + \sum_{j \in T_{int}}\eps(j).\label{within-thm3}
\end{eqnarray}
By construction, every vertex in $T_{int}$ is satisfied by $f'$.
Thus, the $f'$-values correspond to a feasible circulation of
$C(G,T)$.  The cost of the circulation based on the $f'$-values is
\begin{eqnarray*}
\sum_{j \in T_{exp}} \max \left\{0, \left(\sum_{i: (i,j) \in B(T)}
f'(i,j)\right) - 1 \right\} & \leq & \sum_{j \in T_{exp}} \max
\left\{0, \left(\sum_{i: (i,j) \in B(T)} f(i,j)\right) - 1 \right\}
\\ 
&& + \sum_{(u,v) \in B(T)} (f'(u,v) - f(u,v)).
\end{eqnarray*}
We have
\begin{eqnarray*}
\sum_{j \in T_{exp}} \max \left\{0, \left(\sum_{i: (i,j) \in B(T)}
f(i,j)\right) - 1 \right\} & 
\leq & \sum_{j \in T_{exp}} \max\{0,~c_f(j)\} + \sum_{j \in
  T_{exp}}\eps(j).
\end{eqnarray*}
Combining the above inequality with Lemma \ref{lem:fix-f} proves
that the cost of the circulation corresponding to $f'$ is at most \eqref{within-thm3}.
By Lemma \ref{lem:four-five}, we have
\begin{eqnarray*}
\sum_{j \in T_{exp}} \max\{0, c_f(j)\}  + \sum_{j \in T_{int}}
\eps(j) & \leq &  |T_{exp}| \cdot \frac{1}{4} + \sum_{j \in T_{int}} \eps(j)\\
& \leq & \frac{n}{8} + 2\eps n,
\end{eqnarray*}
where the last inequality follows from Fact \ref{exp:bound}.
\qed

\subsection{Combining the $x$- and the $f$-circulations}\label{sec:combine}

We can classify each vertex $j$ in $T_{exp}$ according to the value of
$\xij$.  Intuitively, if many vertices contribute a lot, say $1/3$ to
the $x$-circulation, then they will not contribute a lot to the
$f$-circulation, and vice versa.  The entries in the table below
follow from Claims \ref{clm:end1} and
\ref{clm:end2} and Lemma \ref{clm:fourteen}.
\vspace{2mm}
\begin{center}
\fbox{
\begin{tabular}{c|c|c}
$\xij$ & $c_x(j)$ & $c_f(j)$\\ \hline
$[0,1/4]$ & $\xij/2 $ & $\xij$ \\                
$[1/4,1/2]$ & $\xij/2$ & $1/2-\xij$ \\
$[1/2,1]$ & $\leq 1/3$ & 0 
\end{tabular}
}
\end{center}

\begin{theorem}\label{thm:f2}
When $G$ is subquartic and \optLP $= (1+\eps)n$, there is a
feasible circulation for $C(G,T)$ with cost at most $|B(T)|/11 + 2\eps
n \leq n/11 + 2\eps n + 2$.
\end{theorem}

\begin{proof}
We can compute the cost of the $x$-circulation and the cost of the
$f$-circulation for $C(G,T)$.  We will show that the minimum of the
two costs is upper bounded by the guarantee in the theorem.
In other words, our goal is to show
that for some $\alpha \in [0,1]$, the
following inequality holds:
\begin{eqnarray}
\alpha \sum_{j \in T_{int}} c_x(j) + (1 - \alpha)
\sum_{j \in T_{int}}c_f(j) & \leq & \frac{n}{11}.\label{inequality:alpha}
\end{eqnarray}
Let $\beta \in [0,1]$ represent the fraction of vertices in $T_{exp}$
for which $\xij \in [0,1/2]$.  We may assume without loss of
generality that $x_{min}(j) \geq \frac{1}{4}$ for all $j$; if
$x_{min}(j) < \frac{1}{4}$, we can substitute $x_{min}(j)$ with $1/2 -
x_{min}(j)$, which will preserve the value of $c_f(j)$ and increase
the value of $c_x(j)$ (i.e., we only make the left-hand side of
\eqref{inequality:alpha} higher).  Let $(1-\beta)$ be the remaining
fraction of the vertices, for which $\xij \in (1/2,1]$.

Let $\bar{x}_{min}$ denote the average value of $\xij$ for the
$\beta$-fraction of the vertices in $T_{exp}$ with $\xij \in
[1/4,1/2]$.  Note that $\beta \cdot \bar{x}_{min}/2$ is the average
contribution of these vertices to the $x$-circulation and that $\beta
(1/2 - \bar{x}_{min})$ is the average contribution of these vertices
to the $f$-circulation.  We can take the following convex combination
of the $x$- and $f$-circulations to obtain the following inequality:
\begin{eqnarray}
\frac{6}{11} \left(\beta \cdot \frac{\bar{x}_{min}}{2} + (1-\beta)
\cdot \frac{1}{3}\right) + \frac{5}{11}\left(\beta(\frac{1}{2} -
\bar{x}_{min}) \right) & \leq & \frac{2}{11}.\label{eq:balance}
\end{eqnarray}
Since \eqref{eq:balance} holds when $\bar{x}_{min} \geq 1/4$,
and $\xij \in [1/4,1]$ by assumption, we can conclude that the average
contribution of a vertex in $T_{exp}$ to the circulation is at most
$2/11$.  By Fact \ref{exp:bound}, since there are at most $|B(T)|/2$
vertices in $T_{exp}$, the worst-case cost of the circulation (not
including the contribution of the root) is $|B(T)|/11 + 2\eps n \leq
n/11 + 2\eps n$.
\end{proof}

\section{The $h$-circulation}\label{sec:h-circ}

In this section, we define a third circulation using an optimal
extreme point $x^*$ of \eqref{LP} as defined in Corollary \ref{cor:x}.
First, we fix a greedy DFS tree $T^*$ chosen with respect to $x^*$.
Then we assign each back edge in $B(T^*)$ a value of $\frac{1}{2}$,
which we will refer to as the $h$-values.  If each vertex is satisfied
by $h$, then we have found a circulation with zero cost.  Otherwise,
for each tree cut that is not satisfied by $h$, we increase the
$h$-value on the back edge covering this tree cut to 1.  Note that
these are exactly the tree cuts that are covered by a single back
edge.  More formally, we have the following definitions.
\begin{definition}
With respect to a DFS tree $T^*$, we call a tree cut {\em poor}
if it is covered by only one back edge.  Moreover, we call the
respective back
edge {\em costly}.
\end{definition}
The costly back edges are the ones with $h$-value 1.  All other back
edges have $h$-value $\frac{1}{2}$.  Then the cost of the circulation
is at most
\begin{eqnarray}
\frac{1}{2} \cdot (\text{number of costly back edges}).
\end{eqnarray}

Recall that $H(x^*)$ is the set of vertices that are heavy with
respect to $x^*$.  Let $k = |H(x^*)|$.  We define $\eps^*(v) =
x^*(\delta(v))-2$.

\begin{lemma}\label{lem:h-circ}
When $G$ is subquartic and \optLP $= (1+\eps)n$, there is a feasible
circulation for $C(G,T^*)$ with cost (not including the contribution
of the root) at most $k + 2\eps n$.
\end{lemma}

Since $|E| \leq 2n - k - 1$, observe that $|B(T^*)| \leq n - k$.  To
prove Lemma \ref{lem:h-circ}, we will relate the quantity $k$ to the
number of costly back edges.  We remark that while $x$ and $x^*$ have
the same support, the heavy sets $H(x)$ and $H(x^*)$ might not be the same
(and might not have the same cardinality).  We find it easier to
relate the number of costly back edges to the number of heavy vertices
in $H(x^*)$.  Before we prove Lemma \ref{lem:h-circ}, we show how it can
be used to obtain an improved bound on graph-TSP in subquartic graphs.

\begin{theorem}\label{thm:plus}
When $G$ is subquartic and \optLP $(1+\eps)n$, there is a feasible
circulation either for $C(G,T)$ or for $C(G,T^*)$ with cost at most $n/12 +
2\eps n + 2$.
\end{theorem}

\begin{proof}
By Theorem \ref{thm:f2}, we see that there is a circulation for
$C(G,T)$ of cost at
most $|B(T)|/11 + 2\eps n + 2$.  Note that $|B(T)| = |B(T^*)| = |E| -
(n-1) = n - k$ (see Corollary \ref{cor:x}).  By Lemma \ref{lem:h-circ}, we also have a
circulation for $C(G,T^*)$ of cost at most $k + 2 \eps n + 2$.  These two quantities are
equal when $k = n/12$.  This proves the theorem: When $k \leq
n/12$, the bound on the circulation for $C(G,T^*)$ established via Lemma
\ref{lem:h-circ} is at most $n/12 + 2 \eps n + 2$.  Alternatively, when $k
> n/12$, we have
\begin{eqnarray*} \frac{|B(T)|}{11} & \leq &  \frac{(n - k)}{11} ~ < ~ 
\frac{(n - n/12)}{11} ~ = ~ n/12,
\end{eqnarray*}
and the bound on the circulation for $C(G,T)$ established via Theorem
\ref{thm:f2} is at most $n/12 + 2 \eps n + 2$.
\end{proof}

\begin{theorem}\label{thm:final2}
The approximation guarantee of the M\"omke-Svensson algorithm on
subquartic graphs is at most $25/18$.
\end{theorem}

\begin{proof}
Applying Lemma \ref{lem:MS} and Theorem \ref{thm:plus}, we can compute
an upper bound on the cost of a TSP tour:
\begin{eqnarray*}
\frac{\frac{4n}{3} + \frac{2}{3}\left(\frac{n}{12} + 2\eps n
  \right)}{(1+\eps)n} ~ \leq ~ \frac{\frac{25}{18} + \frac{4\eps}{3}}{(1+\eps)} ~ \leq ~\frac{25}{18}.
\end{eqnarray*}
\end{proof}

It remains to prove Lemma \ref{lem:h-circ}.  Let $k_0$ denote the
number of vertices for which $0 < \eps^*(v) < 1$, and let $k_1$ denote
the number of vertices for which $\eps^*(v) \geq 1$.  Then $k = k_0 +
k_1$.  Additionally, we have
\begin{eqnarray}
2\eps n & = & \sum_{v \in V} \eps^*(v) ~ \geq ~ k_1.
\end{eqnarray}
The following lemma implies Lemma \ref{lem:h-circ}.
\begin{lemma}\label{lem:kk}
When $G$ is
subquartic and \optLP $= (1+\eps)n$, there is a feasible circulation
for $C(G,T^*)$ with cost at most $k_0 + \frac{3}{2}k_1 \leq k + k_1$.
\end{lemma}

\begin{proof}

We set up the following scheme to account for the costly back edges.
\begin{definition}
We say a vertex $j$ is {\em directly charged} for a costly back edge
$f=(u,v)$ if $(j,t) \in E(T^*)$ corresponds to a poor tree cut covered
by $f$ and vertex $j$ is charged for $f$.
\end{definition}

\begin{definition}
We say a vertex $v$ is {\em indirectly charged} for a costly back edge
$f=(u,v)$ if vertex $v$ is charged for $f$.
\end{definition}

Let $f = (u,v)$ be a costly back edge in $B(T^*)$.  
We use the following rules to charge $f$ to a heavy vertex.

\begin{enumerate}

\item If $x^*_f < 1$, then directly charge edge $f$ to some heavy vertex $j$ such that
$(j,t) \in E(T^*)$ and $(j,t)$ corresponds to a poor tree cut covered by $f$.

\item If $x^*_f \geq 1$, then indirectly charge edge $f$ to vertex $v$.

\end{enumerate}

Now we will show that for each vertex $j$, exactly one of the
following statements holds, which shows 
that $C(G,T^*)$ has a circulation cost of at most $k + k_1$.
\begin{enumerate}[(i)]

\item Vertex $j$ is not charged.

\item Vertex $j$ is charged once or twice and $\eps^*(j) > 0$.

\item Vertex $j$ is charged three times and $\eps^*(j) \geq 1$.
\end{enumerate}

\begin{claim}\label{clm:h-heavy} 
If vertex $j$ is charged at least once, then $\eps^*(j) > 0$.
\end{claim}

\begin{cproof}
If $j$ is directly charged for back edge $f = (u,v)$, then $f$ is the
single back edge covering the poor tree cut corresponding to an
outgoing tree edge.  Since $x^*_f < 1$, this tree cut is not satisfied
by $x^*$.  Thus, by Lemma \ref{lem:tight-vertex}, we can conclude that 
$j \in H(x^*)$, which implies $x^*(j) > 2$ and
$\eps^*(j) > 0$.  If $j$ is indirectly charged, then $j$ is incident
to at least two edges with $x^*$-value at least 1, and so $\eps^*(j) >
0$.\end{cproof}

\begin{claim}\label{clm:threetimes}
A vertex $j$ is charged at most three times.
\end{claim}

\begin{cproof}
A vertex $j$ can be directly charged at most once for each outgoing
tree edge and can be indirectly charged for each incoming back edge.
In a subquartic graph, the number of outgoing tree edges plus the
number of incoming back edges is at most three.
\end{cproof}

\begin{claim}\label{clm:threetimes-large-eps}
If a vertex $j$ is charged three times, then $\eps^*(j) \geq 1$.
\end{claim}

\begin{cproof}
If $j$ is not a branch vertex, then the only way that $j$ can be
charged three times is if $j$ is indirectly charged twice and directly
charged once.  If a vertex $j$ is indirectly charged twice, it must
have two incoming back edges with $x^*$-value at least 1.  Thus,
$\eps^*(j) > 1$.

Now consider the case in which vertex $j$ is a branch vertex with
three outgoing tree edges.  Let $e_1, e_2$ and $e_3$ denote three
edges in $E(T^*)$ outgoing from vertex $j$.  Since $j$ cannot be
indirectly charged in this case, it must be directly charged three
times.  Let $f_i \in B(T^*)$ denote the lone back edge covering the
tree cut corresponding to edge $e_i$.  Then we have $x^*_{f_i} +
x^*_{e_i} \geq 2$ (since $e_i$ and $f_i$ cross a cut in $G$) and
$x^*_{f_i} < 1$ implying $x^*_{e_i} > 1$ for $i \in \{1,2,3\}$.  It
follows that $\eps^*(j) > 1$.

The last case to consider is when vertex $j$ is a branch vertex with
two outgoing tree edges $e_1$ and $e_2$.  In this case, if $j$ is
charged three times, then it is directly charged twice and indirectly
charged once.  Assume that the back edge $b$ coming into vertex $j$
comes from the subtree rooted at $t_1$, where $e_1 = (j,t_1)$.  Then
vertex $j$ is indirectly charged for $b$, which implies that $x^*_b
\geq 1$ and $x^*_{e_1} \geq 1$ (since $T^*$ is a greedy DFS tree).
Suppose $f_2$ is the back edge covering the poor tree cut
corresponding to $e_2$.  Then $x^*_{f_2} + x^*_{e_2} \geq 2$ (since
$e_2$ and $f_2$ cross a cut in $G$) and $x^*_{f_2} < 1$, implying that
$x^*_{e_2} > 1$.  Therefore, $x^*_{e_2} + x^*_{e_1} + x^*_b > 3$,
which implies $\eps^*(j) > 1$.
\end{cproof}
This concludes the proof of Lemma \ref{lem:kk}.
\end{proof}

\section{Conclusions}

It is possible that the M\"omke-Svensson algorithm or some close
variant is a $\frac{4}{3}$-approximation algorithm for graph-TSP.  It
is also possible that some graph with maximum degree five could
demonstrate that the approximation ratio of $\frac{13}{9}$ is tight,
although we believe this to be unlikely.  A key step in the algorithm
is to find a minimum cost circulation of the network $C(G,T)$, which
is an easy problem (i.e., it can be solved efficiently), but the key
difficulty is to relate the minimum cost of this circulation to the
minimum cost of a tour.

Previous work used solutions for \eqref{LP} to bound the minimum cost of
a circulation.  In this paper, we explored two new approaches to
analyze this cost.  In the first approach, we showed how to round the
LP values resulting in an improved circulation cost in subquartic
graphs.  Our approach does not immediately extend to general graphs or
even to graphs with maximum degree five, because we try to find a
circulation with values on the back edges close to $1/2$, whereas this
is the worst case for graphs with maximum degree five as shown by
Mucha~\cite{mucha2014frac}.  However, perhaps a different scheme for
modifying the LP values could be designed for such instances.

The second approach is based on extreme point structure: an extreme
point solution for \eqref{LP} with $k$ heavy vertices has at most $n-k$
back edges with respect to any DFS tree.  This is true for general
graphs, and having fewer back edges leads directly to an improved bound
on the circulation cost via the analysis of Mucha.  In the case of
subquartic graphs, we are able to balance this with another analysis
that bounds the cost of a circulation in terms of $k$.  In general
graphs, this latter step may also be possible, but it seems that new
ideas are needed.

Finally, we remark that besides the flexibility in analyzing the
circulation cost, which we have explored in this paper, there is also
the possibility to choose the DFS tree more strategically in order to
obtain an improved analysis of the algorithm.

\section*{Acknowledgements}
We wish to thank Sylvia Boyd, Satoru Iwata, R. Ravi, Andr\'as Seb\H{o}
and Ola Svensson for helpful discussions and comments.  We also thank
the anonymous referees for their detailed remarks that substantially
improved the presentation of the paper.  This work was done
in part while the author was a member of the THL2 group at EPFL.

\bibliography{tsp}
\end{document}